\documentclass[twoside,leqno]{article}
\usepackage[letterpaper]{geometry}
\usepackage{ltexpprt}
\usepackage{cite}
\usepackage{caption}
\usepackage{subcaption}
\usepackage{float}
\usepackage{chngcntr}
\usepackage{apptools}
\usepackage{cancel}
\usepackage{graphicx}
\usepackage{amsmath}
\usepackage{amssymb}
\usepackage{amsfonts}
\usepackage{mathrsfs}
\usepackage{thmtools}
\usepackage{url}
\usepackage{color}
\usepackage{hyperref}
\usepackage{enumitem}
\hypersetup{colorlinks=true,linkcolor=[rgb]{0,0,0},citecolor=[rgb]{0,0,0},urlcolor=black}

\allowdisplaybreaks

\newcommand{\floor}[1]{\left\lfloor #1\right\rfloor}
\newcommand{\ceil}[1]{\left\lceil #1\right\rceil}
\newcommand{\ang}[1]{\left\langle #1 \right\rangle}

\newcommand{\sizeof}[1]{\vert #1 \vert}
\newcommand{\RE}{\mathbb{R}}
\newcommand{\XX}{\mathcal{X}}
\newcommand{\eps}{\varepsilon}

\newcommand{\ST}{\,:\,}
\DeclareMathOperator{\diam}{diam}

\DeclareMathOperator{\vol}{vol}

\DeclareMathOperator{\interior}{int}

\newcommand{\inv}[1]{\frac{1}{#1}}

\newcommand{\bd}{\partial}

\begin{document}

\title{Convex Approximation and the Hilbert Geometry}

\author{%
    Ahmed Abdelkader\thanks{%
        Google LLC, 
        Mountain View, California 94043,
        USA;
        \texttt{abdelkader@google.com}. 
        Research conducted while at the University of Maryland, College Park.
    }	  
    \and
    David M. Mount\thanks{%
	Department of Computer Science and 
	Institute for Advanced Computer Studies,
	University of Maryland,
	College Park, Maryland 20742,
        USA;
	\texttt{mount@umd.edu}.
    }
}
\date{}

\maketitle

\begin{abstract}
The efficient representation of convex bodies in multi-dimensional spaces is a fundamental problem in computational geometry. Several key developments were recently brought about using a number of constructions utilizing Macbeath regions. In this paper, we present a novel intrinsic approach for approximate membership testing, where we carry out the entire development based on structures derived from the Hilbert metric associated with a convex body $K$ in $\RE^d$. First, we revisit the construction of economical Delone sets, deriving the size bound based on the notion of volume entropy. Second, we design a new query structure based on a simple covering by ellipsoids, where queries are answered by ray shooting. As an added bonus, the intrinsic viewpoint facilitates finger searching, where the query time can be bounded by the distance traveled in the Hilbert metric.

\medskip
\noindent\textbf{Keywords:} Convex bodies, Hilbert geometry, Macbeath regions, Membership testing.
\end{abstract}

\maketitle

\section{Introduction} \label{intro.sec}
Imre B\'{a}r\'{a}ny famously quipped, \emph{``Everything interesting that can happen to a convex body happens near its boundary''}~\cite{Bar07}. Evidently, many works on the representation and approximation of convex bodies concern their boundaries. A classical result by Dudley~\cite{Dud74} established that, for any convex body $K$ in $\RE^d$, it is possible to construct an $\eps$-approximating polytope under the Hausdorff distance with $O(1/\eps^{(d-1)/2})$ facets. The facets defining the Dudley approximation are obtained using a uniform sample on an enclosing sphere, which is then projected onto the boundary $\partial K$ yielding a collection of tangent hyperplanes. The same result was also achieved by the dual construction of Bronshteyn and Ivanov~\cite{BrI76}.

Much earlier than Dudley's classical result, Hilbert defined a metric on the \emph{interior} of a convex body~\cite{Hil95}. The Hilbert metric (defined in Section~\ref{hilbert.sec}) has a number of nice properties. It is invariant under projective transformations, and when the body is a Euclidean ball, it generalizes the Cayley-Klein model of hyperbolic geometry. (See the handbook on Hilbert geometry by Papadopoulos and Troyanov~\cite{PaT14}.) This intrinsic point of view provides new insights into classical questions from convexity theory. Despite its natural appeal, there has been little work on the Hilbert geometry in the algorithm design community.

In this paper, we adopt the intrinsic point of view in revisiting the problem of polytope membership testing. The significance of this problem is highlighted by its critical role in several recent developments in geometric approximation algorithms (see, e.g., \cite{AFM17a, AFM17c, AbM18, AAFM20}). Polytope membership queries, both exact and approximate, arise in many application areas, such as linear programming and ray-shooting queries~\cite{Mat93a,MaS93,Cha96a,Ram00}, nearest neighbor searching and the computation of extreme points~\cite{Cla94,Cha96b,AHV04}, collision detection~\cite{EGS99}, and machine learning~\cite{Bur98}.

It is well known that polytope membership is dual to answering halfspace emptiness queries. For $n$ points in $\RE^d$, the fastest exact data structure with near-linear space for $d > 3$ has a query time of roughly $O\big(n^{1-1/\floor{d/2}}\big)$~\cite{Mat92}, which is prohibitive for many applications. Hence, it is natural to consider the problem in an approximate setting.

\begin{figure}[htbp]
\centering
  \includegraphics[scale=0.4]{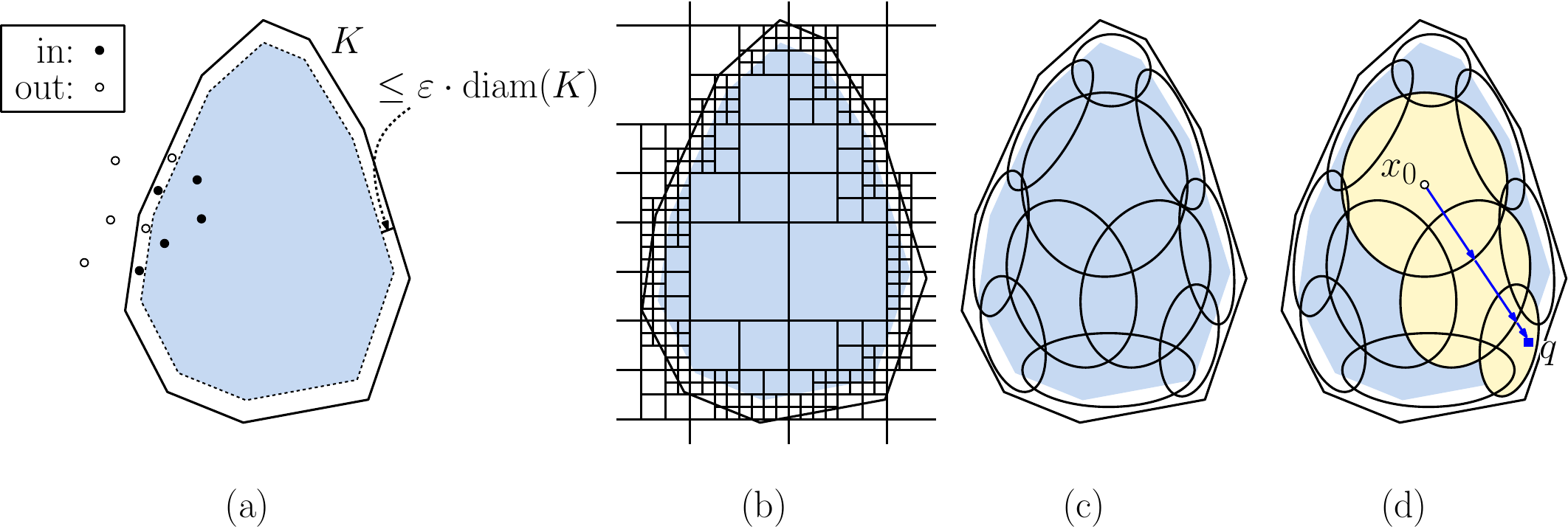}
  \caption{Approximate Polytope Membership (APM).}
  \label{apx-membership.fig}
\end{figure}

Let $\eps$ be a positive real parameter, and let $\diam(K)$ denote $K$'s diameter. Given a query point $q \in \RE^d$, an \emph{$\eps$-approximate polytope membership query} ($\eps$-APM query) returns a negative result if $q \notin K$, a positive result if $q \in K$ and its distance from the boundary of $K$ is greater than $\eps \cdot \diam(K)$, and it may return either result otherwise; see Figure~\ref{apx-membership.fig}(a). (Such an \emph{inner approximation} is a bit more convenient in our framework, but our approach can easily be adapted to answer APM queries with respect to an \emph{outer approximation}.) Our objective is to build an efficient data structure that can answer such queries in $O(\log \inv{\eps})$ time with a storage matching the optimal bound by Dudley.

Recent progress on $\eps$-APM queries arose from the application of Macbeath regions~\cite{Mac52}. In contrast to more traditional space partitioning techniques, such as quadtrees (see Figure~\ref{apx-membership.fig}(b)), Macbeath regions provide shape-sensitive covers enabling for the first time an optimal approach with $O(\log 1/\eps)$ query time and $O(1/\eps^{(d-1)/2})$ storage~\cite{AFM17a}. The data structure of~\cite{AFM17a} is based on a hierarchy of rings providing progressively better approximations of the boundary, where queries are answered by shooting rays from a center point and tracing intersections with the rings.

Later on, the incorporation of the Hilbert metric was explored in~\cite{AbM18}, allowing a more straightforward solution with the same optimality guarantees. The data structure of~\cite{AbM18} disposes of the ray shooting employed in~\cite{AFM17a} in favor of a simple hierarchical covering of the entire convex body (see Figure~\ref{apx-membership.fig}(c)). Intuitively, the hierarchical cover proposed in~\cite{AbM18} is based on progressively finer approximations, where queries are answered by descending through the hierarchy, tracing smaller covers containing the query point. Notably, those covers were related in a precise sense to Delone sets in the Hilbert metric, enabled by the intimate relationship between Macbeath regions and metric balls in the Hilbert geometry. A Delone set consists of a set of points that have useful packing and covering properties with respect to metric balls. However, the great potential of the Hilbert geometry was not fully utilized. Mainly, critical proofs reverted to the machinery of Macbeath regions, hindering further insights on convex approximation.

In this paper, we combine the best of the two approaches~\cite{AFM17a, AbM18} by presenting a streamlined and self-contained development of an optimal data structure for $\eps$-APM queries. The key idea is to replace the hierarchical layer-based structure with a purely \emph{flat} structure based on covering the convex body with ellipsoids. As in \cite{AFM17a}, APM queries are answered by shooting a ray from the origin towards the query point, but in our structure ray-shooting involves simply \emph{walking} the ray through the ellipsoids of the cover, see Figure~\ref{apx-membership.fig}(d). We effectively combine the ray shooting approach from~\cite{AFM17a} with the covering approach of~\cite{AbM18}.

We take the development a step further, obtaining a new result related to finger searching (see Lemma~\ref{apm-finger-search.lem}). In the field of data structures, a \emph{finger search} employs a pointer to a location within a data structure, called a \emph{finger}, and the query time is sensitive to the distance between the location of the finger and the query point. (See Brodal's survey on this topic \cite{Bro05}.) As an application, suppose that rather than performing a single APM query, we are interested in performing a series of $m$ queries along a path, $q_1, \ldots q_m$, where each point $q_{i-1}$ is close to its successor $q_i$. Notably, this scenario is fundamental to random walks in convex bodies where multiple \emph{move proposals} within a given neighborhood are checked at each step~\cite{KLM97,RH09}. While a sequence of $m$ queries could be answered in total time $O(m \log \inv{\eps})$ by querying each point individually, finger search provides a faster solution. Rather than starting each search from scratch, the search for $q_i$ is initiated where the latest search $q_{i-1}$ left off. We show that our new APM data structure supports efficient finger searches, where the time taken to answer an APM query for a point $q_i$ given a finger for point $q_{i-1}$ is $O(1 + d_K(q_{i-1},q_i))$, where $d_K$ denotes the Hilbert distance relative to the body $K$. (See Section~\ref{hilbert.sec} for definitions.) We also show that if the query points lie within an $\eps$-erosion of $K$, the Hilbert distance between any two of them is never larger than $O(\log \inv{\eps})$. This implies that finger searching is never slower, and potentially much faster, than answering queries individually.

Together, our main result is summarized as follows.

\begin{theorem} \label{main.thm}
Given a convex body $K$ and $\eps > 0$, there exists a data structure of space $O(1/\eps^{(d-1)/2})$ that answers $\eps$-approximate polytope membership queries in time $O(\log 1/\eps)$. The data structure supports \emph{finger searching} in $O(1 + \Delta)$, where $\Delta$ is the distance between consecutive query points in the Hilbert metric defined on $K$, and can be constructed in time $O(n + 1/\eps^d)$ when $K$ is the intersection of $n$ halfspaces.
\end{theorem}

The rest of the paper is organized as follows. In Section~\ref{prelim.sec}, we present basic definitions and preliminary results. Then, in Section~\ref{macbeath-delone.sec}, we present a self-contained development of the Macbeath-based Delone sets in the Hilbert metric, first presented in~\cite{AbM18}, where we introduce a new and more direct proof of the upper bound on the cardinality of the Delone set. In Section~\ref{apm.sec}, we proceed to apply the Delone set to design a new data structure, combining and significantly simplifying the two prior approaches from~\cite{AFM17a} and~\cite{AbM18}. Finally, we conclude in Section~\ref{conc.sec}.

\section{Preliminaries} \label{prelim.sec}

We consider real $d$-dimensional space, $\RE^d$, where $d$ is a fixed constant. Throughout, we use $K$ to denote a convex body in $\RE^d$. Let $\bd K$ and $\interior K$ denote $K$'s boundary and interior, respectively. Our space and query-time bounds hold irrespective of $K$'s representation, e.g., as the intersection of $n$ closed halfspaces. (The choice of representation will only affect the preprocessing times.)

Let $O$ denote the origin of $\RE^d$. For $\delta > 0$, let $B(x,\delta)$ denote the Euclidean ball of radius $\delta$ centered at $x \in \RE^d$, and let $B(\delta) = B(O,\delta)$. Given a vector $v \in \RE^d$, let $\|v\|$ denote its Euclidean length. 

\subsection{The Hilbert Metric.} \label{hilbert.sec}
The Hilbert metric~\cite{Hil95} was introduced over a century ago by David Hilbert. We review the definition and basic properties we use for our data structure; see the survey by Papadopoulos and Troyanov \cite{PaT14} and the multimedia contribution by Nielsen and Shao~\cite{NiS17}.

A \emph{Hilbert metric} $(K,d_K)$ consists of a convex body $K$ in $\RE^d$ and a distance function defined on its points. For any pair of distinct points $x, y \in K$, let $x'$ and $y'$ denote the endpoints of the chord of $K$ through $x$ and $y$, labeled so the points appear in the order $\ang{x', x, y, y'}$ (see Figure~\ref{hilbert.fig}(a)). Using $|x y|$ as shorthand for the Euclidean distance $\|x - y\|$, the Hilbert distance $d_K$ is defined as
\[
  d_K(x, y)
	~ = ~ \frac{1}{2} \ln \left( \frac{| x' y |}{| x' x |} \frac{| x y' |}{| y y' |} \right).
\]
The Hilbert distance $d_K$ is a metric and line segments are geodesics. Observe that if we fix either $x$ or $y$ and move the other to $\partial K$, $d_K(x,y)$ increases to $\infty$.

\begin{figure}[htbp]
\centering
  \includegraphics[scale=0.4]{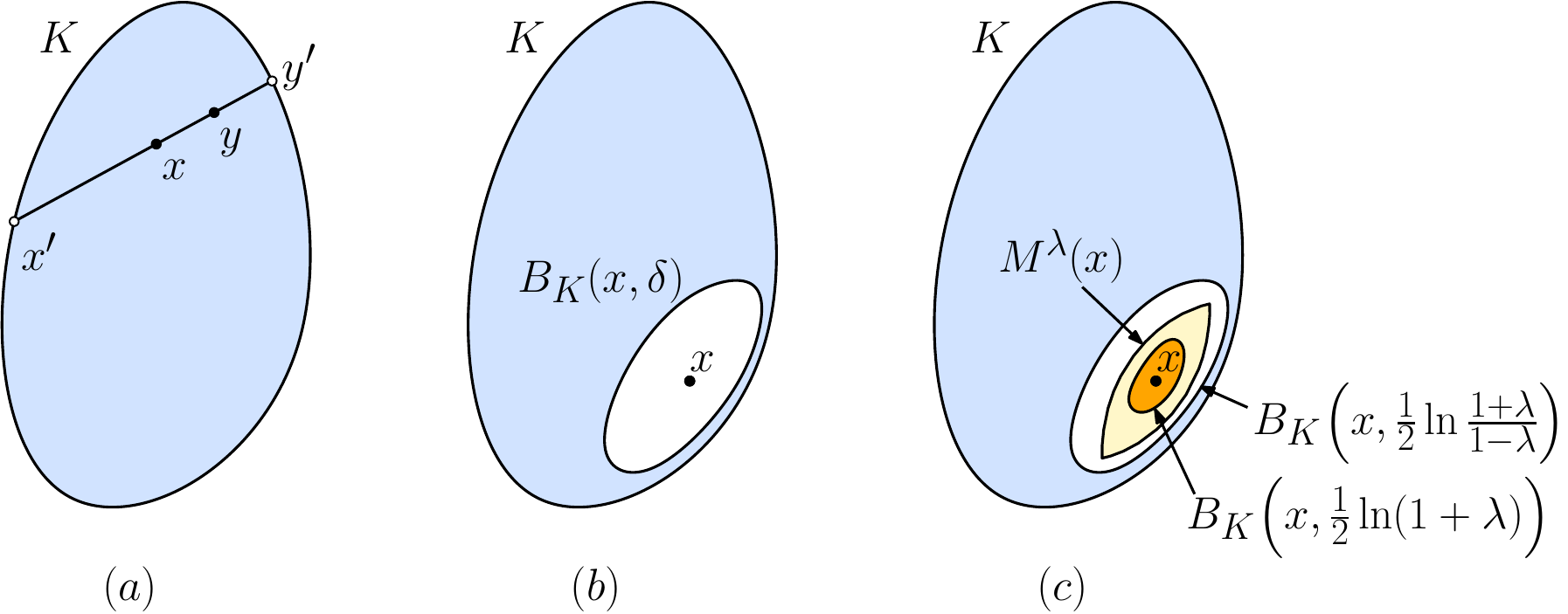}
  \caption{(a) The Hilbert metric, (b) Hilbert metric ball, and (c) Lemma~\ref{macbeath-to-hilbert.lem}.}
  \label{hilbert.fig}
\end{figure}

The Hilbert distance has a number of interesting properties. It is invariant under projective transformations, and when $K$ is a unit ball, the Hilbert distance is equal (up to a constant factor) to the distance between points in the Cayley-Klein model of hyperbolic geometry. Given a point $x \in K$ and $\delta > 0$, let 
\[
    B_K(x,\delta) ~ = ~ \{ y \in K \ST d_K(x, y) \leq \delta \}
\]
denote the \emph{Hilbert ball} of radius $\delta$ about $x$ (see Figure~\ref{hilbert.fig}(b)). Hilbert balls are convex. When $K$ is a polytope with $n$ facets in $\RE^d$, a Hilbert ball $B_K$ is a polytope with at most $n(n-1)$ facets~\cite{GeM21}, and the bound improves to $2 n$ when $d = 2$~\cite{NiS17}.

\subsection{Preconditioning and Perturbation.} \label{canonical_perturbation.sec}
For $0 < \gamma \le 1$, we say that $K$ is in \emph{$\gamma$-canonical form} if $B(\gamma/2) \subseteq K \subseteq B(1/2)$ (see Figure~\ref{prelim.fig}(a)). Clearly, such a body has a diameter between $\gamma$ and $1$. We say that a body is in \emph{canonical form} if it is in $\gamma$-canonical form for some constant $\gamma$ (possibly depending on the dimension).

\begin{figure}[htbp]
\centering
  \includegraphics[scale=0.4]{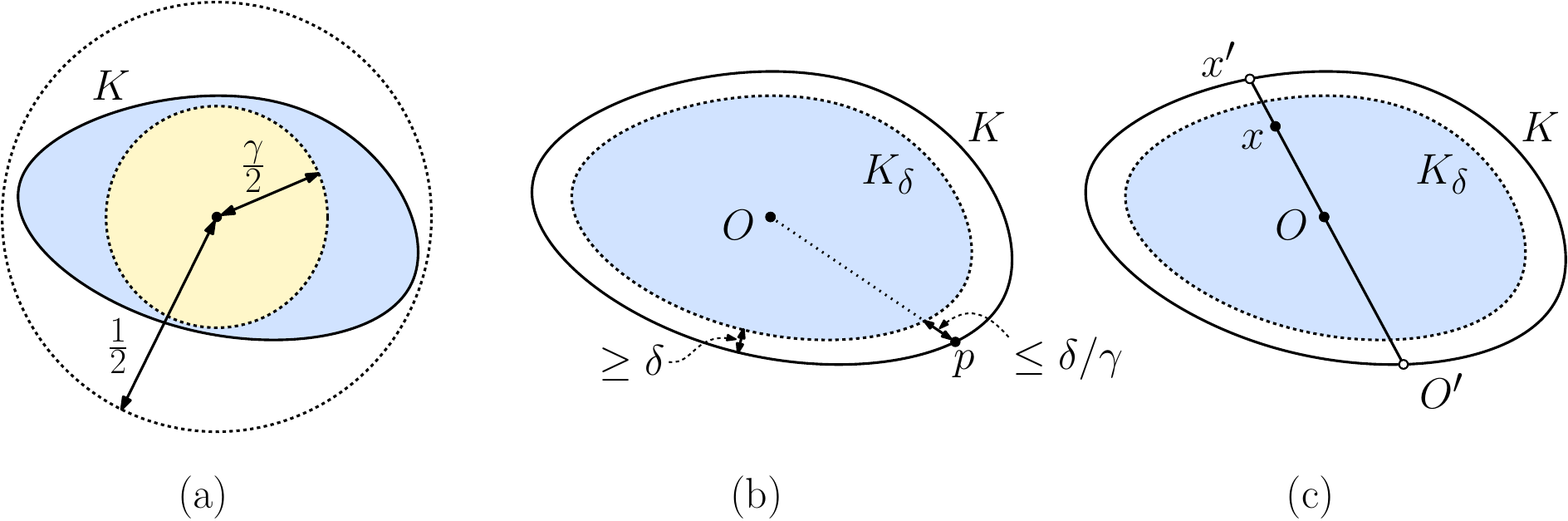}
  \caption{(a) A convex body $K$ in $\gamma$-canonical form, (b) the $\delta$-erosion $K_{\delta}$, and (c) Lemma~\ref{hilbert-radius.lem}.}
  \label{prelim.fig}
\end{figure}

It is possible to compute, in $O(n)$ time, a non-singular affine transformation $T$ such that $T(K)$ is in $(1/d)$-canonical form~\cite{Har01, AFM17c}. Further, if a convex body $P$ is within Hausdorff distance $\eps$ of $T(K)$, then $T^{-1}(P)$ is within Hausdorff distance at most $d \kern+1pt \eps$ of $K$. Therefore, for the sake of approximation with respect to Hausdorff distance, we may assume that $K$ has been mapped to canonical form, and $\eps$ is scaled by a factor of $1/d$. Because we assume $d$ is a constant, this transformation only affects the constant factors in our analysis. As a consequence, the resulting approximation is directionally sensitive, being more accurate along directions where $K$ is skinnier (see, e.g., \cite{AFM17b}).

A number of our constructions involve perturbing the body $K$ by means of a small contraction. Given a convex body $K$ in $\gamma$-canonical form, and a parameter $\delta > 0$, define $K$'s $\delta$-erosion as the scaled copy $K_{\delta} = (1 - 2 \delta/\gamma)K$. Because every point of $K$ is at distance at least $\gamma/2$ and at most $1/2$ from the origin, the following result is immediate~\cite{AAFM20} (see Figure~\ref{prelim.fig}(b)).

\begin{lemma} \label{erosion_gap.lem}
The Hausdorff distance between $K$ and $K_{\delta}$ is at least $\delta$ and at most $\delta/\gamma$.
\end{lemma}

To further characterize the errors due to the perturbation in terms of the Hilbert metric, we show that the Hilbert ball of radius $\inv{2} \ln \inv{\delta} + O(1)$ centered at the origin is nested between $K$ and $K_{\delta}$.

\begin{lemma}\label{hilbert-radius.lem}
Given a convex body $K$ in $\gamma$-canonical form and $\delta > 0$, 
\[
    K_\delta 
        ~ \subseteq ~ B_K\kern-2pt\left( O, \, \inv{2} \ln \inv{\delta\gamma} \right) 
        ~ \subseteq ~ K.
\]
\end{lemma}

\begin{proof}
The second inclusion is trivial, since $B_K(O, r) \subset K$ for any finite $r$. To prove the first inclusion, consider any $x \in K_{\delta} = (1 - 2 \delta/\gamma)K$, where $x \neq O$. Let $O'$ and $x'$ be the endpoints on $\bd K$ of the chord passing through $O$ and $x$, as in the definition of the Hilbert distance (see Figure~\ref{prelim.fig}(c)). Since $x \in K_{\delta}$, we have $\|x' - x\| \geq \delta$. Since $K$ is in $\gamma$-canonical form, we have $\|O' - x\| \leq 1$, $\|O' - O\| \geq \frac{\gamma}{2}$, and $\|x' - O\| \leq \frac{1}{2}$. Therefore,
\[
    d_K(O, x)
	~ =    ~ \frac{1}{2} \ln \left( \frac{\|O' - x\|}{\|x' - x\|} \frac{\|x' - O\|}{\|O' - O\|} \right)
        ~ \leq ~ \frac{1}{2} \ln \left( \frac{1}{\delta} \cdot \frac{1/2}{\gamma/2} \right)
        ~ =    ~ \frac{1}{2} \ln \frac{1}{\delta\gamma}.
\]
Assuming $\gamma$ is a constant, this is $\inv{2} \ln \inv{\delta} + O(1)$.
\end{proof}

\subsection{Macbeath Regions.} \label{macbeath.sec}
Our algorithms and data structures will involve packings and coverings by ellipsoids. These ellipsoids will be based on a classical concept from convexity theory, called \emph{Macbeath regions}, which were described first by A.\ M.\ Macbeath \cite{Mac52}. They have found uses in diverse areas (see, e.g., B{\'a}r{\'a}ny's survey~\cite{Bar00}). Given a convex body $K$, a point $x \in K$, and a real parameter $\lambda \ge 0$, the \emph{$\lambda$-scaled Macbeath region} at $x$, denoted $M^{\lambda}_K(x)$, is defined to be
\[
  x + \lambda ((K-x) \cap (x-K)).
\]
When $\lambda = 1$, it is easy to verify that $M^{1}_K(x)$ is the intersection of $K$ and the reflection of $K$ around $x$ (see Figure~\ref{macbeath_john.fig}(a)), and hence it is centrally symmetric about $x$. $M^{\lambda}_K(x)$ is a scaled copy of $M^{1}_K(x)$ by the factor $\lambda$ about $x$. We refer to $x$ and $\lambda$ as the \emph{center} and \emph{scaling factor} of $M^{\lambda}_K(x)$, respectively. To simplify the notation, when $K$ is clear from the context, we often omit explicit reference in the subscript. When $\lambda < 1$, we say $M^{\lambda}(x)$ is \emph{shrunken}.

\begin{figure}[htbp]
\centering
\begin{subfigure}[b]{0.42\columnwidth}\centering
  \includegraphics[scale=0.42]{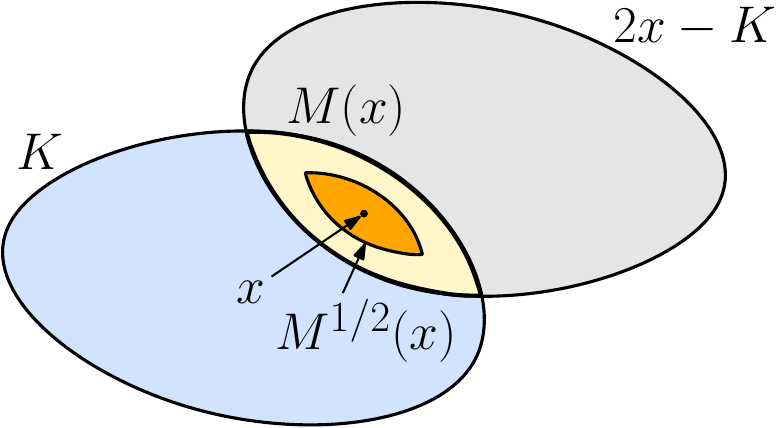}
  \caption{}
\end{subfigure}
\begin{subfigure}[b]{0.42\columnwidth}\centering
  \includegraphics[scale=0.42]{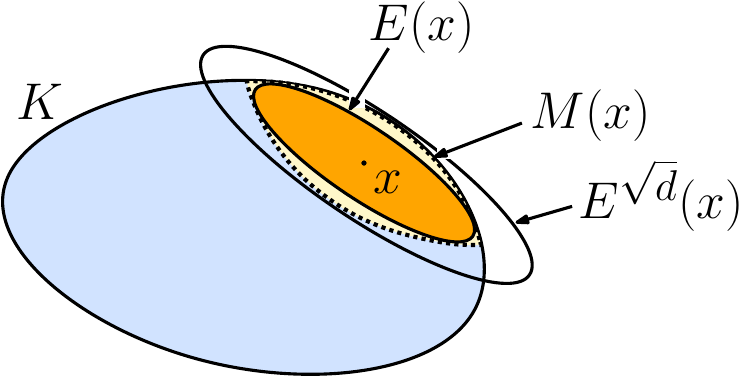}
  \caption{}
\end{subfigure}
  \caption{(a) Macbeath regions and (b) Macbeath ellipsoids.}
  \label{macbeath_john.fig}
\end{figure}

To provide some intuition as to why this is related to the Hilbert geometry consider any chord of $K$, and (recalling Figure~\ref{hilbert.fig}(a)) place two points $x$ and $y$ on this chord such that $x$ is closer to $y'$ than to $x'$, and $y$ is at the midpoint of $x y'$. It is easy to see that $y$ lies on the boundary of $M^{1/2}_K(x)$. It is also at constant Hilbert distance from $x$ (ranging from $\ln \sqrt{2}$ to $\ln \sqrt{3}$ depending on how close $x$ is to $y'$). In general, a constant-factor shrinking of a Macbeath region corresponds roughly a to Hilbert ball of constant radius (see Figure~\ref{hilbert.fig}(c)). This is established formally in the following lemma from~\cite{AbM18}, which generalizes a result of Vernicos and Walsh~\cite{VeW16} to arbitrary values of $\lambda$. For example, for $\lambda = 1/5$, we have $B_K(x, 0.09) \subseteq M^{1/5}(x) \subseteq B_K(x, 0.21)$ for all $x \in \interior(K)$. The proof involves straightforward manipulations of the definitions and is provided in Appendix~\ref{deferred.sec} for completeness.

\begin{restatable}{lemma}{MacbeathToHilbert} \label{macbeath-to-hilbert.lem}
Given a convex body $K \subset \RE^d$, for all $x \in \interior(K)$ and any $0 \leq \lambda < 1$,
\[
  B_K\kern-3pt\left( x, \frac{1}{2}\ln{(1 + \lambda)} \right) 
	~ \subseteq ~ M^{\lambda}(x) 
	~ \subseteq ~ B_K\kern-3pt\left( x, \frac{1}{2}\ln{\frac{1 + \lambda}{1 - \lambda}} \right) .
\]
\end{restatable}

For the sake of efficiency, it will be useful to approximate Macbeath regions by shapes of constant combinatorial complexity. Given a Macbeath region $M_K(x)$, define its associated \emph{Macbeath ellipsoid} $E_K(x)$ as the maximum-volume inscribed ellipsoid (see Figure~\ref{macbeath_john.fig}(b)). Clearly, $E_K(x)$ is centered at $x$, and hence, can be defined equivalently as the maximum-volume ellipsoid centered at $x$ within $K$. It is well known that the maximum-volume ellipsoid contained within a convex body is unique, and it can be computed for a convex polytope in time linear in the number of its bounding halfspaces~\cite{ChM96}.

Let $E^{\lambda}_K(x)$ denote an $\lambda$-factor scaling of $E_K(x)$ about $x$. By John's Theorem (applied in the context of centrally symmetric bodies) it follows that $E^{\lambda}_K(x) \subseteq M^{\lambda}_K(x) \subseteq E^{\lambda\sqrt{d}}_K(x)$ \cite{Bal97}. As a direct consequence, we have the following adaptation of Lemma~\ref{macbeath-to-hilbert.lem}.

\begin{corollary}\label{macbeath-to-hilbert.cor}
Given a convex body $K \subset \RE^d$, for all $x \in \interior(K)$ and any $0 \leq \lambda < 1$,
\[
  B_K\kern-3pt\left( x, \frac{1}{2}\ln \left(1 + \frac{\lambda}{\sqrt{d}} \right) \right) 
	~ \subseteq ~ E^{\lambda}_K(x) 
	~ \subseteq ~ B_K\kern-3pt\left( x, \frac{1}{2}\ln{\frac{1 + \lambda}{1 - \lambda}} \right).
\]
\end{corollary}

\section{Delone Sets in the Hilbert Metric} \label{macbeath-delone.sec}

In this section, we develop an intrinsic approach to the approximation of convex bodies based on Macbeath ellipsoids. Let us assume that the convex body $K$ is in canonical form.\footnote{This is merely a convenience that allows us to use absolute errors. Since Macbeath regions and ellipsoids are affine invariants, our constructions apply even if $K$ is not provided in canonical form.} Recall from Section~\ref{prelim.sec} that for $\eps > 0$, $K_{\eps}$ is a contraction of $K$ such that the boundaries of $K$ and $K_{\eps}$ are separated by a distance of $\eps$.

Our approximation will be based in part on covering $K$ with a collection of Macbeath ellipsoids. The efficiency of the data structure relies on properties of these ellipsoids that are reminiscent of Delone sets in metric spaces. In this section, we review this concept, which arises from packing and covering by metric balls~\cite{Sut09,Cla06}. The concept of covering all or part of a convex body through the use of ``ball-like'' convex shapes has been presented before in the form of $(2,\eps)$-coverings~\cite{NaV22} and MNets~\cite{AFM23}.

Given a metric $f$ over a domain $\XX$, a point $x \in \XX$, and real $r > 0$, define the ball $B_{f}(x, r) = \{ y \in \XX : f(x,y) \le r\}$. For $\eps, \eps_p, \eps_c > 0$, a set $X \subseteq \XX$ is an:
\begin{description}
\vspace{5pt}\setlength{\itemsep}{2pt}%
\item[$\eps$-packing:] If the balls of radius $\eps/2$ centered at every point of $X$ do not intersect.

\item[$\eps$-covering:] If every point of $\XX$ is within distance $\eps$ of some point of $X$.

\item[$(\eps_p, \eps_c)$-Delone Set:] If $X$ is an $\eps_p$-packing and an $\eps_c$-covering.
\end{description}
Delone sets (also known as \emph{nets}) have been used in the design of data structures for answering geometric proximity queries in metric spaces through the use of hierarchies of nets, such as navigating nets \cite{KrL04}, net trees \cite{HaM06}, and cover trees \cite{BKL06}. (While a Delone set is a set of points, we will often blur the distinction between the points and the associated pairs of packing and covering balls centered at those points.)

\subsection{Expansion-Containment.} \label{expansion-containment.sec}

An important property of Macbeath regions, which we call \emph{expansion-containment}, is that if two shrunken Macbeath regions overlap, then an appropriate expansion of one contains the other. The following is a generalization of results of Ewald, Rogers and Larman \cite{ELR70} and Br\"{o}nnimann, Chazelle, and Pach~\cite{BCP93}. Our generalization allows the shrinking factor to be adjusted, and also shows how to adjust the expansion factor of the first body when it suffices to cover a scaled version of the second body, e.g., the center point only. The proof follows by straightforward algebra and is reproduced in Appendix~\ref{deferred.sec} for completeness.

\begin{figure}[htbp]
\centering
  \includegraphics[scale=0.40]{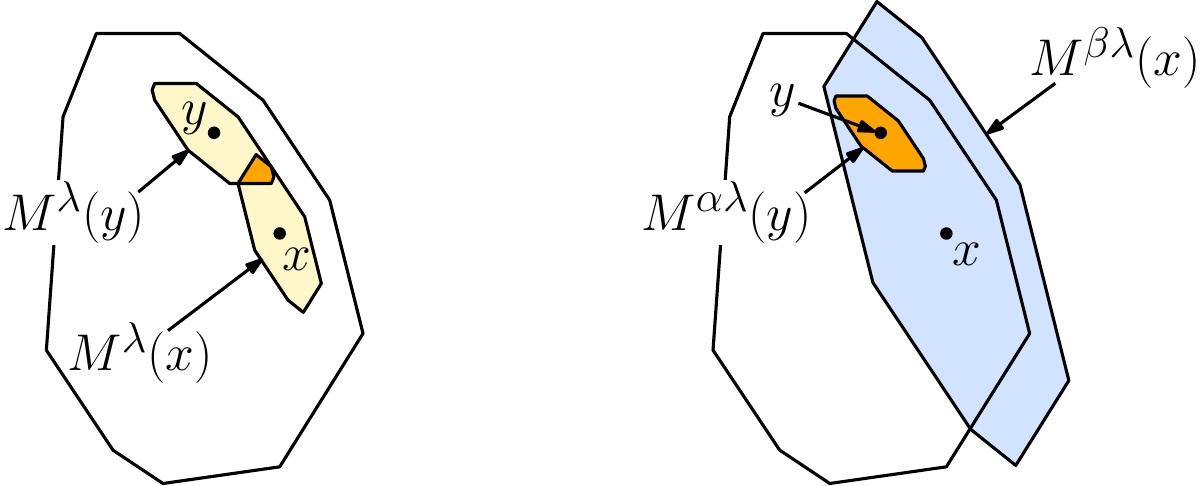}
  \caption{Expansion-containment per Lemma~\ref{exp-con.lem}.}
  \label{exp-con.fig}
\end{figure}

\begin{restatable}{lemma}{ExpansionContainment} \label{exp-con.lem}
\emph{(Expansion-Containment for Macbeath Regions)}
Let $K \subset \RE^d$ be a convex body and let $0 < \lambda < 1$. If $x,y \in K$ such that $M^{\lambda}(x) \cap M^{\lambda}(y) \neq \emptyset$, then for any $\alpha \ge 0$ and $\beta = \frac{2 + \alpha(1+\lambda)}{1 - \lambda}$, $M^{\alpha\lambda}(y) \subseteq M^{\beta\lambda}(x)$ (see Figure~\ref{exp-con.fig}).
\end{restatable}

By John's theorem, there is an easy extension to Macbeath ellipsoids.

\begin{corollary}[Expansion-Containment for Macbeath Ellipsoids]\label{exp-con-ellipse.lem}
Let $K \subset \RE^d$ be a convex body and let $0 < \lambda < 1$. If $x,y \in K$ such that $E^{\lambda}(x) \cap E^{\lambda}(y) \neq \emptyset$, then for any $\alpha \ge 0$ and $\beta = \frac{2 + \alpha(1+\lambda)}{1 - \lambda}\sqrt{d}$, $E^{\alpha\lambda}(y) \subseteq E^{\beta\lambda}(x)$.
\end{corollary}

\subsection{Realizing Delone Sets.} \label{delone.sec}

For the sake of building a Delone set, we employ two scaling factors, one for packing and one for covering. Define constants $\lambda_p = \inv{8\sqrt{d}}$ and $\lambda_c = \inv{2}$. For any $x \in K$, let $E''(x) = E^{\lambda_p}(x)$ and $E'(x) = E^{\lambda_c}(x)$ (see Figure~\ref{covering-ellipsoids.fig}(a)). The next lemma from~\cite{AbM18} shows that we can create a Delone set from any maximal packing of the $E''$ ellipsoids. (See Appendix~\ref{deferred.sec} for the proof).

\begin{figure}[htbp]
\centering
  \includegraphics[scale=0.38]{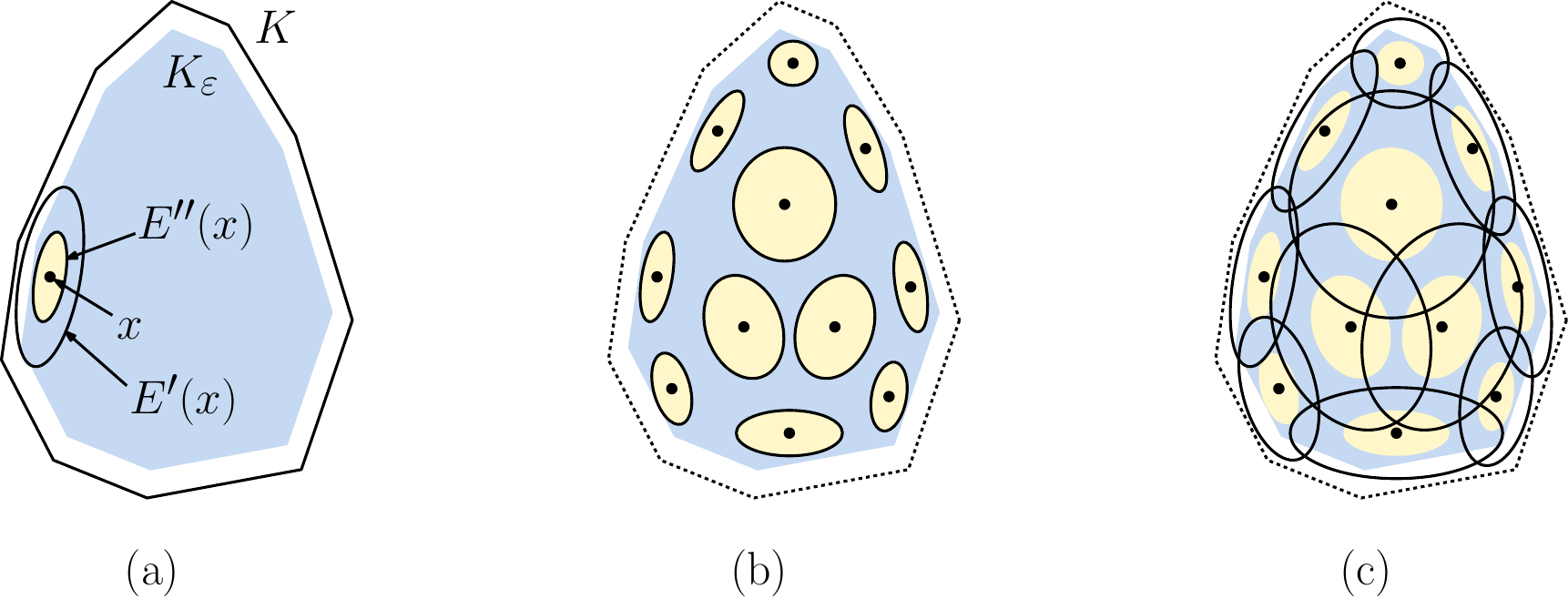}
  \caption{A Delone set for a convex body based on Macbeath ellipsoids. (Not drawn to scale.)}
  \label{covering-ellipsoids.fig}
\end{figure}

\begin{restatable}{lemma}{PackingCovering} \label{packing-covering.lem}
\emph{(Macbeath-Based Delone Set)}
Given a convex body $K \subset \RE^d$ and $\eps > 0$, let $X_{\eps}(K) = X_{\eps}$ denote any maximal set of points in $K_{\eps}$ such that the ellipsoids $E''(x)$ (defined above) are pairwise disjoint (see Figure~\ref{covering-ellipsoids.fig}(b)). Then $X_{\eps}$ is a Delone set over $K_{\eps}$ in the sense that the $E''$ ellipsoids form a packing and the $E'$ ellipsoids form a covering (see Figure~\ref{covering-ellipsoids.fig}(c)) of $K_{\eps}$. Furthermore, $\bigcup_{x \in X_{\eps}} E'(x) \subseteq K$.
\end{restatable}

Henceforth, we use the term \emph{Delone set} for a given convex body $K$ and $\eps > 0$ in the above sense, namely as any maximal set of points in $K_{\eps}$ such that the ellipsoids $E''(x)$ are pairwise disjoint. 

\subsection{Size Bounds.} \label{delone-size.sec}

Let us now bound the size of the Delone set $X_{\eps}$ for any convex body $K$. We first observe the following easy upper bound. Because the boundaries of $K$ and $K_{\eps}$ are at least $\eps$ apart, for each $x \in K_{\eps}$ the Euclidean ball $B(x,\eps)$ is contained within $K$, and hence it lies within the Macbeath region $M(x)$. It follows that the volume of $E''(x)$ is $\Omega(\vol(B(x,\eps))) = \Theta(\eps^d)$. As the ellipsoids $E''$ are disjoint and $K$'s volume is a constant, a packing argument shows that $\sizeof{X_\eps}$ is $O(1/\eps^d)$.

The upper bound from simple packing is clearly inferior to the bound suggested by Dudley's result. Prior works utilizing Macbeath regions~\cite{AFM17a,AbM18} established the desired bound, obtaining variants of Theorem~\ref{size-bound.thm} below, using the width-based economical cap cover \cite{AFM17c,AAFM20}, which is a convenient adaptation of B{\'a}r{\'a}ny's volume-based economical cap cover~\cite{Bar00}. We will present a more direct proof based solely on the growth rate of the volume measure in the Hilbert geometry.

\begin{theorem} \label{size-bound.thm}
For a convex body $K \subset \RE^d$ and $\eps > 0$, any Delone set $X_{\eps}$ over $K_{\eps}$ has $\sizeof{X_\eps} = O(1/\eps^{(d-1)/2})$.
\end{theorem}

\begin{proof}
Our proof is based on the connection between Macbeath regions and Hilbert balls (from Lemma~\ref{macbeath-to-hilbert.cor}) combined with the upper bound on the volume entropy of the Hilbert metric. It has long been conjectured that the growth rate of the volume of metric balls in the Hilbert geometry, as the radius tends to infinity, never exceeds that of the hyperbolic geometry. This was recently proved in full generality by Benoist and Hulin~\cite{BeH13} and Tholozan~\cite{Tho17}, providing the following bound on the volume of $B_K(O, r)$.
\[
    \vol\big(B_K(O, r)\big) 
        ~ = ~ O\big(e^{(d-1)r}\big).
\]

By Lemma~\ref{hilbert-radius.lem}, $K_{\eps/2}$ is contained within a Hilbert ball of radius $r_{\eps} = \inv{2} \ln \frac{2}{\eps\gamma}$. Recall that the ellipsoids $\{E''(x)\}_{x \in X_\eps}$ are pairwise disjoint. By Corollary~\ref{macbeath-to-hilbert.cor}, each such ellipsoid $E''(x)$ contains a Hilbert ball of radius $\inv{2}\ln \left(1 + \frac{\lambda_p}{d}\right)$, and since those balls are pairwise disjoint, their number is upper bounded by the volume of the Hilbert ball $B_K(O,r_{\eps})$, which by the volume bound cited above is at most
\[
    \exp((d-1)r_{\eps})
        ~ = ~ \exp\left( \frac{d-1}{2} \ln \frac{2}{\eps\gamma} \right)
        ~ = ~ \left( \frac{2}{\eps\gamma} \right)^{\frac{d-1}{2}}
        ~ = ~ O\left( \left( \inv{\eps} \right)^{\frac{d-1}{2}}\right),
\]
with the hidden constant depending on $d$.
\end{proof}

\section{Approximate Polytope Membership (APM)} \label{apm.sec}

In this section we present our Delone-based approach to answering $\eps$-APM queries for a convex body $K$ in $\RE^d$ in canonical form. Recall that such a query must answer ``no'' if $x \notin K$, ``yes'' if $x \in K_{\eps}$, and may give either response otherwise. Recall the Delone-set $X_{\eps}$ of Section~\ref{delone.sec}, that is, any maximal subset of $K_{\eps}$ such that the ellipsoids $E''$ form a packing. Our APM data structure consists simply of the intersection graph of the covering ellipsoids $E'$ of $X_{\eps}$ (see Figure~\ref{apm-graph.fig}(a), (b)). More formally, define $G_{\eps}(K)$ to be an undirected graph whose vertex set is $X_{\eps}$, and whose edges are pairs $x,y \in X_{\eps}$ such that $E'(x) \cap E'(y) \neq \emptyset$. This graph will have a distinguished vertex $x_0 \in X_{\eps}$, called the \emph{root}, such that $E'(x_0)$ contains the origin. (If there are multiple ellipsoids covering the origin, take any one.)

\begin{figure}[htbp]
  \centerline{\includegraphics[scale=0.40]{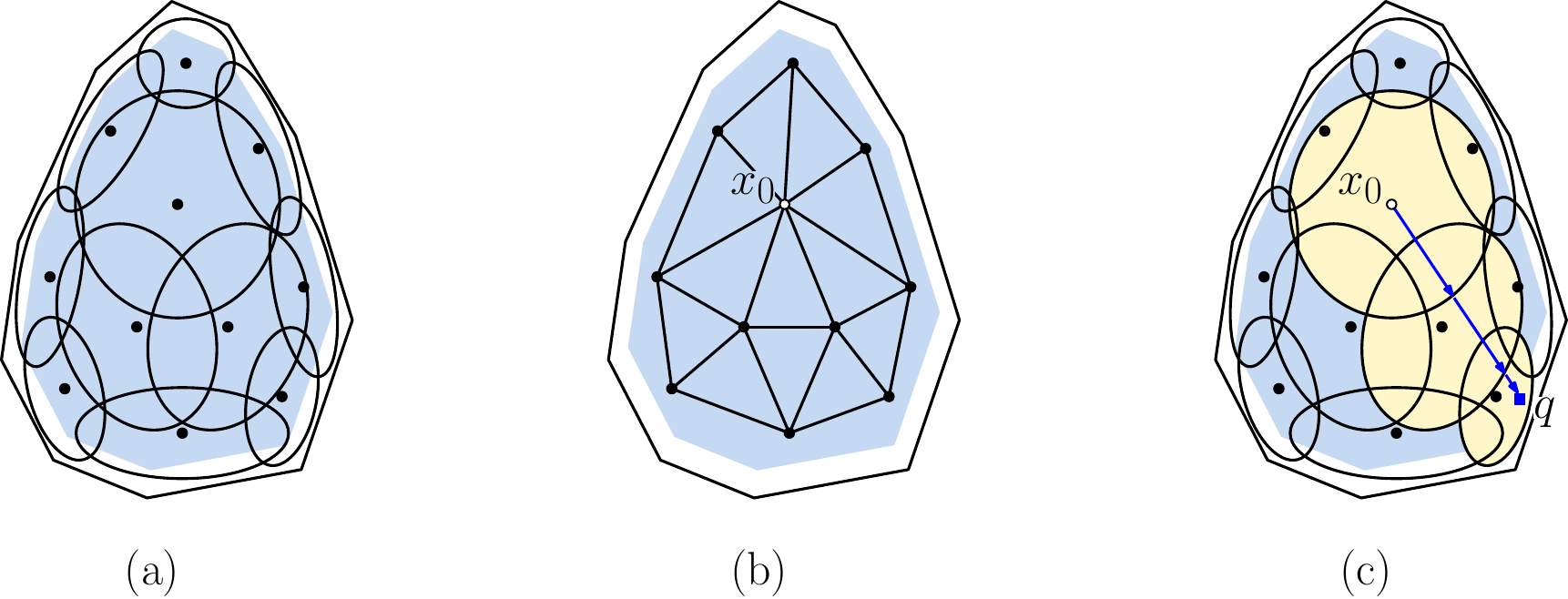}}
  \caption{Data structure for answering APM queries.} \label{apm-graph.fig}
\end{figure}

To answer an APM query for a point $q$, we walk along the ray from the origin to $q$, passing through the ellipsoids that intersect this ray (see Figure~\ref{apm-graph.fig}(c)). We start with the root by setting $x \gets x_0$. If $q \in E'(x)$ then we know that $q \in K$, and it is safe to answer ``yes.'' Otherwise, we inspect the neighbors of $x$ in the intersection graph. If we cannot progress further along the ray by traveling through one of the ellipsoids associated with these neighbors, then we stop and report ``no.'' This is justified because the extreme point $p$ where the ray exits $E'(x)$ lies on the boundary of $\bigcup_{z \in X_{\eps}} E'(z)$. Since these ellipsoids cover $K_{\eps}$, every point on this boundary is external to $K_{\eps}$, and it is safe to answer ``no.'' Otherwise, we determine the neighbor $y$ of $x$ in $G_{\eps}(K)$ for which the ray can travel the farthest through $E'(y)$. We set $x \gets y$ and continue the search. (For example, in Figure~\ref{apm-graph.fig}(c) the search visits the three highlighted ellipsoids and reports that $q$ lies approximately within $K$.)

The storage needed by this data structure is dominated by the size of the graph $G_{\eps}(K)$. The query time is the product of the maximum number of ellipsoids visited in the search and the maximum degree of any vertex in $G_{\eps}(K)$. The next few lemmas bound each of these quantities.

First, we show that the graph has constant degree. This is a standard property of Delone sets in spaces of constant dimension. It follows by combining the expansion-containment property with a simple packing argument. The proof appears in Appendix~\ref{deferred.sec} for completeness.

\begin{restatable}{lemma}{APMOverlapBound} \label{apm-overlap-bound.lem}
The maximum degree of the graph $G_{\eps}(K)$ is $O(1)$.
\end{restatable}

Next, we analyze the number of ellipsoids visited while answering any query. Our next lemma is a useful technical result, which states that the number of ellipsoids traversed by the search procedure is roughly proportional to the Hilbert distance traveled. (Recall that, unlike the Euclidean distance, Hilbert distances grow without limit as we approach the boundary of $K$.) Here is an intuitive justification. For any constant $\lambda < 1$, Corollary~\ref{macbeath-to-hilbert.cor} tells us that the shrunken Macbeath $E^{\lambda}(x)$ contains a Hilbert ball of constant radius. Thus, each time we walk the ray through a Macbeath ellipsoid in the search, we are moving at least a constant distance in terms of the Hilbert metric. Therefore, the number of iterations needed by the search algorithm is roughly proportional to the Hilbert distance traveled.

\begin{lemma}[Finger Search] \label{apm-finger-search.lem}
Given two points $p, q \in K_{\eps}$, let $E'(x)$ be any covering ellipsoid containing $p$ with $x \in X_{\eps}$. If our APM search procedure starts at $E'$, it runs in $O(1 + d_K(p, q))$ time.
\end{lemma}

\begin{proof}
Because $G$ has bounded degree, each iteration of the search procedure runs in $O(1)$ time. It suffices to count the number of iterations. Starting the search at $p$, let $x_0, \ldots, x_k$ denote the vertices of $G_{\eps}(K)$ visited by the search (see Figure~\ref{finger-search.fig}(a)), and for $1 \leq i \leq k$, let $p_i$ denote the point where the search procedure exits $E'(x_{i-1})$ and enters $E'(x_i)$. Because straight line segments are geodesics in the Hilbert metric, 
\[
    d_K(p, q) 
        ~ = ~ d_K(p, p_1) + \left(\sum_{i=1}^{k-1} d_K(p_i, p_{i+1}) \right) + d_K(p_k,q).
\]
We assert that for $1 \leq i \leq k-1$, $d_K(p_i, p_{i+1})$ is bounded below by some constant $c$ (depending on the dimension but independent of $\eps$). It follows that the number of iterations is at most $2 + \ceil{d_K(p_1, p_k)/c} = O(1 + d_K(p, q))$.

\begin{figure}[htbp]
  \centerline{\includegraphics[scale=0.40]{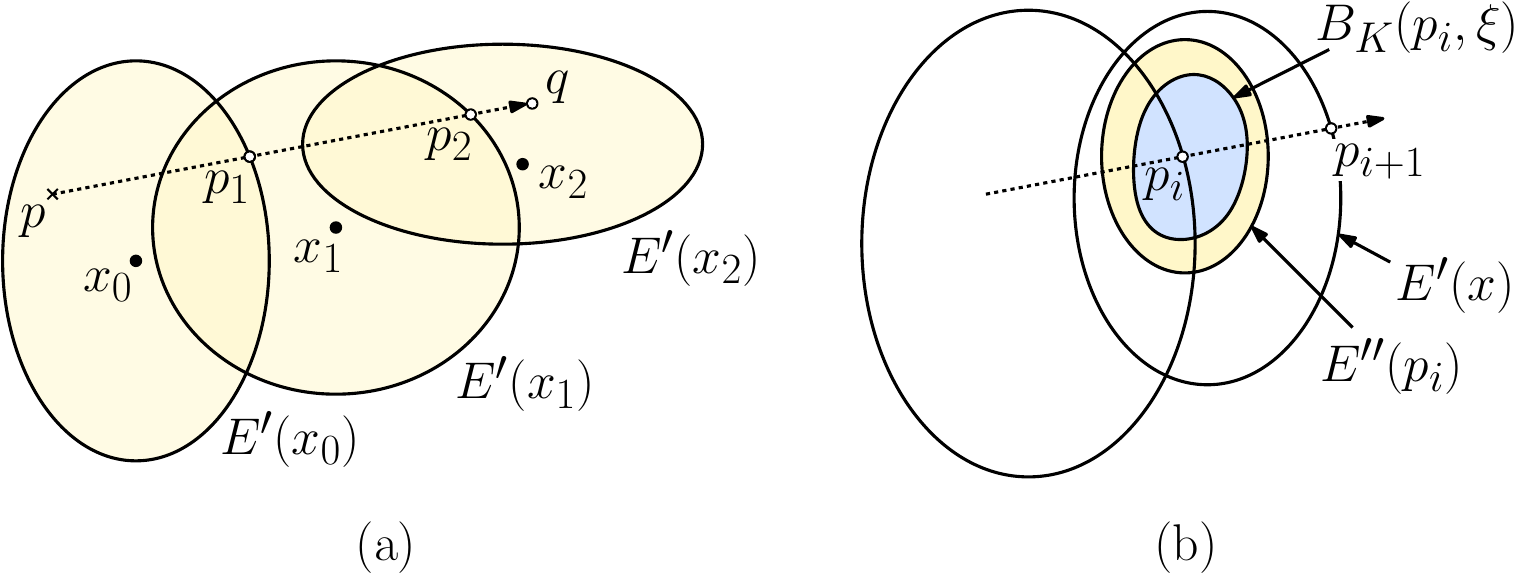}}
  \caption{Proof of Lemma~\ref{apm-finger-search.lem}.} \label{finger-search.fig}
\end{figure}

To establish the assertion, let $c = \frac{1}{2}\ln (1 + \lambda_p/\sqrt{d})$. Consider any point $p_i$ in this sequence. Since $p$ and $q$ lie in $K_{\eps}$, by convexity all the intermediate points on the ray do as well. In the proof of Lemma~\ref{packing-covering.lem}, we showed that for any $y \in K_{\eps}$, there exists $x \in X_{\eps}$ such that $E''(y) \subset E'(x)$. Therefore, there is a point $x \in X_{\eps}$ such that $E''(p_i) \subset E'(x)$ (see Figure~\ref{finger-search.fig}(b)). By Corollary~\ref{macbeath-to-hilbert.cor}, we have
\[
    B_K(p_i, c)  
        ~ \subseteq ~ E^{\lambda_p}(p_i)
        ~ =         ~ E''(p_i)
        ~ \subset   ~ E'(x).
\]
Clearly, $x$ is a neighbor of $x_{i-1}$ in $G_{\eps}(K)$. Therefore the search will travel along the ray at least until exiting $E'(x)$, implying that it will travel entirely through the above Hilbert ball, progressing a distance of at least $c$, which is clearly a constant. This establishes the assertion and completes the proof.
\end{proof}

Recalling from Lemma~\ref{hilbert-radius.lem} that the distance from the origin to any point of $K_{\eps}$ is $O(\log \inv{\eps})$. The following is a straightforward consequence.

\begin{lemma} \label{apm-query-time.lem}
Given any query point $q \in \RE^d$, the APM search procedure runs in $O(1 + \log \inv{\eps})$ time.
\end{lemma}

\begin{proof}
If $q$ lies in $K_{\eps}$, the result follows directly from Lemmas~\ref{hilbert-radius.lem} and~\ref{apm-finger-search.lem}. If not, let $q'$ denote the point where the ray from to origin to $q$ intersects the boundary of $K_{\eps}$. The lemma bounds the time to walk the ray to $q'$, and the search procedure terminates at most one iteration later, either reporting ``yes'' if $q$ is in the final ellipsoid or reporting ``no'' if not.
\end{proof}

Finally, we show that the storage needed for our data structure is proportional Dudley's bound on the complexity of Hausdorff approximating a convex body by halfspaces. As established in Theorem~\ref{size-bound.thm}, the size of the Delone set $X_{\eps}(K)$ is $O(1/\eps^{(d-1)/2})$. Because the graph has constant degree, and each vertex of the graph stores only a constant amount of additional information (the coordinates of the associated point and the equation of the covering ellipsoid), the same bound applies to the size of the entire data structure.

\begin{lemma} \label{apm-space-bound.lem}
The storage required $G_{\eps}(K)$ is $O(1/\eps^{(d-1)/2})$.
\end{lemma}

To finish the development, we describe a straightforward approach to construct the Macbeath-based Delone set $X_\eps$ as needed for the data structure. We note that obtaining the best dependencies on $\eps$ in the construction time likely involves fairly sophisticated methods (see, e.g., Arya et al.~\cite{AFM17b}). The next lemma is a direct adaptation of a related result in~\cite{AbM18}; see Appendix~\ref{deferred.sec} for the proof.

\begin{restatable}{lemma}{APMConstruction} \label{apm-construction.lem}
Given a convex body $K \subset \RE^d$ represented as the intersection of $n$ halfspaces and $\eps > 0$, $G_{\eps}(K)$ can be computed in time $O(n + 1/\eps^{O(d)})$, where the constant in the exponent does not depend on $\eps$ or $d$.
\end{restatable}

This completes the proof of Theorem~\ref{main.thm}.

\section{Concluding Remarks} \label{conc.sec}
In this paper, we consider the derivation of geometric query structures intrinsic to the natural metrics defined by the data at hand. Leveraging prior work on the approximation of convex polytopes, we present a self-contained development of Delone sets in the Hilbert metric associated with a convex body $K$ in $\RE^d$. The Delone set is realized by a covering of a scaled-down copy of $K$ using a collection of ellipsoids derived from Macbeath regions. We further simplify the proof of the size bound by sidestepping the economic cap cover construction, utilized crucially in prior work, and present a new proof based on the elegant comparison theorem of Bishop exploiting the relationship between different metrics. We proceed to demonstrate the application of the Delone set to the design of an efficient data structure for approximate polytope membership testing, with optimal query time and storage complexity. Queries are answered by shooting rays from the origin and navigating through the ellipsoidal cover. Notably, the intrinsic viewpoint enables finger searching, where the query time can be upper bounded by the distance traveled in the Hilbert metric. For future work, we anticipate further developments of geometric algorithms in the Hilbert metric and further applications of the Delone set construction.


\pdfbookmark[1]{References}{s:ref}
\bibliographystyle{plainurl}
\bibliography{convex}

@string{ANN_MATH = {Ann.\ of Math.}}

@string{CGTA = {Comput.\ Geom.\ Theory Appl.}}

@string{DCG = {Discrete Comput.\ Geom.}}

@string{DMKD = {Data Min.\ Knowl.\ Discov.}}

@string{FOCS_2001 = {Proc.\ 42nd Annu.\ IEEE Sympos.\ Found.\ Comput.\ Sci.}}

@string{J_APPROX_THY = {J.\ Approx.\ Theory}}

@string{JACM = {J.\ Assoc.\ Comput.\ Mach.}}

@string{JALG = {J.\ Algorithms}}

@string{MATHIKA = {Mathematika}}

@string{RCMP = {Rend.\ Circ.\ Mat.\ Palermo}}

@string{SIB_MATH_J = {Siberian Math.\ J.}}

@string{SICOMP = {SIAM J.\ Comput.}}

@string{SOCG_1994 = {Proc.\ Tenth Annu.\ Sympos.\ Comput.\ Geom.}}

@string{SOCG_1996 = {Proc.\ 12th Annu.\ Sympos.\ Comput.\ Geom.}}

@string{SOCG_2000 = {Proc.\ 16th Annu.\ Sympos.\ Comput.\ Geom.}}

@string{SOCG_2017 = {Proc.\ 33rd Internat.\ Sympos.\ Comput.\ Geom.}}

@string{SOCG_2023 = {Proc.\ 39th Internat.\ Sympos.\ Comput.\ Geom.}}

@string{SODA_1999 = {Proc.\ Tenth Annu.\ ACM-SIAM Sympos.\ Discrete Algorithms}}

@string{SODA_2004 = {Proc.\ 15th Annu.\ ACM-SIAM Sympos.\ Discrete Algorithms}}

@string{SODA_2017 = {Proc.\ 28th Annu.\ ACM-SIAM Sympos.\ Discrete Algorithms}}

@string{SODA_2020 = {Proc.\ 31st Annu.\ ACM-SIAM Sympos.\ Discrete Algorithms}}

@string{SODA_2023 = {Proc.\ 34th Annu.\ ACM-SIAM Sympos.\ Discrete Algorithms}}

@string{STOC_2006 = {Proc.\ 38th Annu.\ ACM Sympos.\ Theory Comput.}}

@string{STOC_2009 = {Proc.\ 41st Annu.\ ACM Sympos.\ Theory Comput.}}

@string{SWAT_2018 = {Proc.\ 16th Scand.\ Workshop Algorithm Theory}}

@inproceedings{AbM18,
 author = {A. Abdelkader and D. M. Mount},
 title = {Economical {Delone} sets for approximating convex bodies},
 booktitle = SWAT_2018,
 year = {2018},
 pages = {4:1--4:12},
 doi = {10.4230/LIPIcs.SWAT.2018.4},
}

@article{AHV04,
 author = {P. K. Agarwal and S. Har-Peled and K. R. Varadarajan},
 title = {Approximating extent measures of points},
 journal = JACM,
 volume = {51},
 year = {2004},
 pages = {606--635},
 doi = {10.1145/1008731.1008736},
}

@inproceedings{AAFM20,
 author = {R. Arya and S. Arya and G. D. {da Fonseca} and D. M. Mount},
 title = {Optimal bound on the combinatorial complexity of approximating polytopes},
 booktitle = SODA_2020,
 year = {2020},
 pages = {786--805},
 doi = {10.1137/1.9781611975994.48},
}

@inproceedings{AFM17a,
 author = {S. Arya and G. D. da Fonseca and D. M. Mount},
 title = {Optimal approximate polytope membership},
 booktitle = SODA_2017,
 pages = {270--288},
 year = 2017,
 doi = {10.1137/1.9781611974782.18},
}

@inproceedings{AFM17b,
 author = {S. Arya and G. D. da Fonseca and D. M. Mount},
 title = {Near-optimal $\varepsilon$-kernel construction and related problems},
 booktitle = SOCG_2017,
 pages = {10:1--15},
 year = 2017,
 doi = {10.4230/LIPIcs.SoCG.2017.10},
 url = {https://arxiv.org/abs/1703.10868},
}

@article{AFM17c,
 author = {S. Arya and G. D. da Fonseca and D. M. Mount},
 title = {On the combinatorial complexity of approximating polytopes},
 journal = DCG,
 volume = {58},
 number = {4},
 year = {2017},
 pages = {849--870},
 doi = {10.1007/s00454-016-9856-5},
}

@article{AFM18a,
 author = {S. Arya and G. D. da Fonseca and D. M. Mount},
 title = {Approximate polytope membership queries},
 journal = SICOMP,
 volume = {47},
 number = {1},
 year = {2018},
 pages = {1--51},
 doi = {10.1137/16M1061096},
}

@inproceedings{AFM23,
 author = {S. Arya and G. D. da Fonseca and D. M. Mount},
 title = {Economical Convex Coverings and Applications},
 booktitle = SODA_2023,
 pages = {1834--1861},
 year = 2023,
 doi  = {10.1137/1.9781611977554.ch70},
}

@incollection{Bal97,
 author = {K. Ball},
 title = {An Elementary Introduction to Modern Convex Geometry},
 booktitle = {Flavors of Geometry},
 editor = {S. Levy},
 series = {MSRI Publications},
 publisher = {Cambridge University Press},
 address = {Cambridge, UK},
 volume = {31},
 year = {1997},
 pages = {1--58},
 isbn = {9780521629621},
}

@article{Bar00,
 author = {I. B{\'a}r{\'a}ny},
 title = {The technique of {M}-regions and cap-coverings: {A} survey},
 journal = RCMP,
 volume = {65},
 year = {2000},
 pages = {21--38},
 url = {https://users.renyi.hu/~barany/},
}

@incollection{Bar07,
 author = {I. B{\'a}r{\'a}ny},
 title = {Random polytopes, convex bodies, and approximation},
 booktitle = {Stochastic Geometry},
 editor = {W. Weil},
 series = {Lecture Notes in Mathematics},
 publisher = {Springer},
 volume = {1892},
 year = {2007},
 pages = {77--118},
 doi = {10.1007/978-3-540-38175-4_2},
}

@article{BeH13,
 author = {Y. Benoist and D. Hulin},
 title = {Cubic differentials and finite volume convex projective surfaces},
 journal = {Geom.\ Topol.},
 volume = {17},
 year = {2013},
 pages = {595--620},
 doi = {10.2140/gt.2013.17.595},
}

@inproceedings{BKL06,
 author = {A. Beygelzimer and S. Kakade and J. Langford},
 title = {Cover trees for nearest neighbor},
 booktitle = {Proc.\ 23rd Internat.\ Conf.\ on Mach.\ Learn.},
 year = {2006},
 pages = {97--104},
 doi = {10.1145/1143844.1143857},
}

@incollection{Bro05,
 author = {G. S. Brodal},
 title = {Finger Search Trees},
 booktitle = {Handbook of Data Structures and Applications},
 edition = {Second},
 editors = {D. P. Mehta and S. Sahni},
 publisher = {CRC Press},
 year = {2005},
 doi = {10.1201/9781315119335},
}

@article{BCP93,
 author = {H. Br{\"o}nnimann and B. Chazelle and J. Pach},
 title = {How hard is halfspace range searching?},
 journal = DCG,
 volume = {10},
 year = {1993},
 pages = {143--155},
 doi = {10.1007/BF02573971},
}

@article{BrI76,
 author = {E. M. Bronshteyn and L. D. Ivanov},
 title = {The approximation of convex sets by polyhedra},
 journal = SIB_MATH_J,
 volume = {16},
 year = {1976},
 pages = {852--853},
 doi = {10.1007/BF00967115},
}

@article{Bur98,
 author = {C. J. C. Burges},
 title = {A Tutorial on Support Vector Machines for Pattern Recognition},
 journal = DMKD,
 volume = {2},
 number = {2},
 year = {1998},
 pages = {121--167},
 doi = {10.1023/A:1009715923555},
 issn = {1384-5810},
}

@inproceedings{Cha96a,
 author = {T. M. Chan},
 title = {Fixed-dimensional linear programming queries made easy},
 booktitle = SOCG_1996,
 year = {1996},
 pages = {284--290},
 doi = {10.1145/237218.237397},
}

@article{Cha96b,
 author = {T. M. Chan},
 title = {Output-sensitive results on convex hulls, extreme points, and related problems},
 journal = DCG,
 volume = {16},
 year = {1996},
 pages = {369--387},
 doi = {10.1007/BF02712874},
}

@article{ChM96,
 author = {B. Chazelle and J. Matou\v{s}ek},
 title = {On Linear-Time Deterministic Algorithms for Optimization Problems in Fixed Dimension},
 journal = JALG,
 volume = {21},
 year = {1996},
 pages = {579--597},
 doi = {10.1006/jagm.1996.0060},
}

@inproceedings{Cla94,
 author = {K. L. Clarkson},
 title = {An algorithm for approximate closest-point queries},
 booktitle = SOCG_1994,
 year = {1994},
 pages = {160--164},
 doi = {10.1145/177424.177609},
}

@inproceedings{Cla06,
 author = {K. L. Clarkson},
 title = {Building triangulations using $\varepsilon$-nets},
 booktitle = STOC_2006,
 year = {2006},
 pages = {326--335},
 doi = {10.1145/1132516.1132564},
}

@article{Dud74,
 author = {R. M. Dudley},
 title = {Metric entropy of some classes of sets with differentiable boundaries},
 journal = J_APPROX_THY,
 volume = {10},
 number = {3},
 year = {1974},
 pages = {227--236},
 doi = {10.1016/0021-9045(74)90120-8},
}

@inproceedings{EGS99,
 author = {J. Erickson and L. J. Guibas and J. Stolfi and L. Zhang},
 title = {Separation-sensitive collision detection for convex objects},
 booktitle = SODA_1999,
 year = {1999},
 pages = {327--336},
}

@article{ELR70,
 author = {G. Ewald and D. G. Larman and C. A. Rogers},
 title = {The directions of the line segments and of the $r$-dimensional balls on the boundary of a convex body in {Euclidean} space},
 journal = MATHIKA,
 volume = {17},
 year = {1970},
 pages = {1--20},
 doi = {10.1112/S0025579300002655},
}

@inproceedings{GeM21,
 author = {A. H. Gezalyan and D. M. Mount},
 title = {Voronoi Diagrams in the {Hilbert} Metric},
 booktitle = SOCG_2023,
 pages = {35:1--35:16},
 year = {2023},
 volume = {258},
 doi = {10.4230/LIPIcs.SoCG.2023.35},
}

@inproceedings{Har01,
 author = {S. Har-Peled},
 title = {A replacement for {Voronoi} diagrams of near linear size},
 booktitle = FOCS_2001,
 year = {2001},
 pages = {94--103},
 doi = {10.1109/SFCS.2001.959884},
}

@article{HaM06,
 author = {S. Har-Peled and M. Mendel},
 title = {Fast construction of nets in low dimensional metrics, and their applications},
 journal = SICOMP,
 volume = {35},
 year = {2006},
 pages = {1148--1184},
 doi = {10.1145/1064092.1064117},
}

@article{Hil95,
 author = {D. Hilbert},
 title = {{\"U}ber die gerade {Linie} als k{\" u}rzeste {Verbindung} zweier {Punkte}},
 journal = {Mathematische Annalen},
 volume = {46},
 year = {1895},
 pages = {91--96},
 doi = {10.1007/BF02096204},
}

@article{KLM97,
  author = {R. Kannan and L. Lov{\'a}sz and M. Simonovits},
  title = {Random walks and an {$O^*(n^5)$} volume algorithm for convex bodies},
  journal = {Random Structures \& Algorithms},
  volume = {11},
  pages = {1--50},
  year = {1997},
  doi = {10.1002/(SICI)1098-2418(199708)11:1<1::AID-RSA1>3.0.CO;2-X},
}

@inproceedings{RH09,
 author = {R. Kannan and H. Narayanan},
 title = {Random Walks on Polytopes and an Affine Interior Point Method for Linear Programming},
 booktitle = STOC_2009,
 year = {2009},
 pages = {561--570},
 doi = {10.1145/1536414.1536491},
}

@inproceedings{KrL04,
 author = {R. Krauthgamer and J. R. Lee},
 title = {Navigating nets: {Simple} algorithms for proximity search},
 booktitle = SODA_2004,
 year = {2004},
 pages = {798--807},
 url = {https://dl.acm.org/doi/abs/10.5555/982792.982913},
}

@article{Mac52,
 author = {A. M. Macbeath},
 title = {A theorem on non-homogeneous lattices},
 journal = ANN_MATH,
 volume = {56},
 year = {1952},
 pages = {269--293},
 doi = {10.2307/1969800},
}

@article{Mat92,
 author = {J. Matou\v{s}ek},
 title = {Reporting points in halfspaces},
 journal = CGTA,
 volume = {2},
 issue = {3},
 year = {1992},
 pages = {169--186},
 doi = {10.1016/0925-7721(92)90006-E},
}

@article{Mat93a,
 author = {J. Matou\v{s}ek},
 title = {Linear optimization queries},
 journal = JALG,
 volume = {14},
 number = {3},
 year = {1993},
 pages = {432--448},
 doi = {10.1006/jagm.1993.1023},
}

@article{MaS93,
 author = {J. Matou{\v{s}}ek and O. Schwarzkopf},
 title = {On ray shooting in convex polytopes},
 journal = DCG,
 volume = {10},
 year = {1993},
 pages = {215--232},
 doi = {10.1007/BF02573975},
}

@article{NaV22,
 author = {M. Nasz{\'o}di and M. Venzin},
 title = {Covering convex bodies and the closest vector problem},
 journal = DCG,
 volume = {67},
 year = {2022},
 pages = {1191--1210},
 doi = {10.1007/s00454-022-00392-x},
}

@inproceedings{NiS17,
 author = {F. Nielsen and L. Shao},
 title = {{On Balls in a Hilbert Polygonal Geometry (Multimedia Contribution)}},
 booktitle = SOCG_2017,
 pages = {67:1--67:4},
 year = {2017},
 doi = {10.4230/LIPIcs.SoCG.2017.67},
}

@book{PaT14,
 author = {A. Papadopoulos and M. Troyanov},
 title = {Handbook of {Hilbert} Geometry},
 publisher = {EMS Press},
 year = {2014},
 doi = {10.4171/147},
}

@inproceedings{Ram00,
 author = {E. A. Ramos},
 title = {Linear programming queries revisited},
 booktitle = SOCG_2000,
 year = {2000},
 pages = {176--181},
 doi = {10.1145/336154.336198},
}

@book{Sut09,
 author = {W. A. Sutherland},
 title = {Introduction to metric and topological spaces},
 publisher = {Oxford University Press},
 year = {1975},
 isbn = {9780199563081},
}

@article{Tho17,
 author = {N. Tholozan},
 title = {Volume entropy of {Hilbert} metrics and length spectrum of {Hitchin} representations into {PSL$(3,\mathbb{R})$}},
 journal = {Duke Math.\ J.},
 volume = {166},
 number = {7},
 year = {2017}, 
 pages = {1377--1403},
 doi = {10.1215/00127094-00000010X},
}

@unpublished{VeW16,
 author = {C. Vernicos and C. Walsh},
 title = {Flag-approximability of convex bodies and volume growth of {Hilbert} geometries},
 year = {2016},
 note = {HAL Archive (hal-01423693i)},
 url = {https://hal.archives-ouvertes.fr/hal-01423693},
}


\appendix

\section{Deferred Proofs} \label{deferred.sec}
In this appendix, we present proofs of some technical lemmas, which were deferred from the main body.

\medskip

\MacbeathToHilbert*

\begin{proof}
Let $(x',x,y,y')$ be as in Figure~\ref{hilbert.fig}. We can express $d_K(x, y)$ as $\frac{1}{2} \ln \left(  \frac{\|y' - x\|}{\|y' - y\|} \frac{\|x' - y\|}{\|x' - x\|} \right)$. Since these points are collinear, $y \in M^\lambda(x)$ if and only if
\[
  \max\left( \frac{\|y - x\|}{\|y' - x\|}, \frac{\|y - x\|}{\|x' - x\|} \right) 
	~ \leq ~ \lambda \quad \text{that is,} \quad \frac{\|y - x\|}{\|y' - x\|} \leq \lambda \text{ and } \frac{\|y - x\|}{\|x' - x\|} \leq \lambda.
\]
This corresponds to the following bounds on the ratios used in defining $d_K(x, y)$.
\begin{align*}
  \frac{\|y - x\|}{\|y' - x\|} ~ \leq ~ \lambda
	&	~~ \Leftrightarrow ~~ \|y' - y\| ~ \geq ~ (1 - \lambda) \|y' - x\| 
		~~ \Leftrightarrow ~~  \frac{\|y' - x\|}{\|y' - y\|} ~ \leq ~ \frac{1}{1 - \lambda}, \\
  \frac{\|y - x\|}{\|x' - x\|} ~ \leq ~ \lambda
	&	~~ \Leftrightarrow ~~ \|x' - y\| ~ \leq ~ (1 + \lambda) \|x' - x\| 
		~~ \Leftrightarrow ~~ \frac{\|x' - y\|}{\|x' - x\|} ~ \leq ~ 1 + \lambda.
\end{align*}
We conclude $d_K(x, y) \leq \frac{1}{2} \ln \frac{1 + \lambda}{1 - \lambda}$ for all $y \in M^\lambda(x)$, implying $M^\lambda(x) \subseteq B_K(x, \frac{1}{2}\ln \frac{1 + \lambda}{1 - \lambda})$. On the other hand, observing that $1 + \lambda \leq \inv{1-\lambda}$, any $y \in K$ satisfying $d_K(x, y) = \frac{1}{2} \ln(1 + \lambda)$ clearly lies in $M^\lambda(x)$ implying $B_K(x, \frac{1}{2}\ln (1 + \lambda)) \subseteq M^\lambda(x)$.
\end{proof}

\ExpansionContainment*

\begin{proof}
By definition of the scaled Macbeath region, we can express any $z \in M^{\lambda}(x) \cap M^{\lambda}(y)$ as $z = y + \lambda (k_1 - y) = x + \lambda(x - k_2)$, for some $k_1, k_2 \in K$. Hence, we have
\[
  y 
	~ = ~ \frac{1 + \lambda}{1 - \lambda} x - \frac{\lambda}{1 - \lambda}(k_1 + k_2).
\]
It suffices to show that if $v \in M^{\alpha\lambda}(y)$ then $v \in M^{\beta\lambda}(x)$. To see this, observe first that for some $k_3 \in K$, we can express $v$ as
\begin{eqnarray*}
  v
	& = & y + \alpha\lambda(y - k_3)
	~ = ~ (1 + \alpha\lambda)y - \alpha\lambda \kern+1pt k_3 \\
	& = & (1 + \alpha\lambda) \left( \frac{1 + \lambda}{1 - \lambda} x - \frac{\lambda}{1 - \lambda}(k_1 + k_2) \right) - \alpha\lambda \kern+1pt k_3.
\end{eqnarray*}
We can express the coefficient of $x$ as
\[
  \frac{(1 + \alpha \lambda)(1 + \lambda)}{1 - \lambda}
	~ = ~ 1 + \frac{2\lambda + \alpha\lambda(1+\lambda)}{1 - \lambda}
	~ = ~ 1 + \beta\lambda,
\]
and we can express the negation of the sum of the $k_i$ terms as
\begin{eqnarray*}
  \frac{\lambda}{1 - \lambda} ( (1 + \alpha\lambda) (k_1 + k_2) + \alpha(1 - \lambda) k_3 )
        & = & \beta\lambda \left( \frac{1 + \alpha\lambda}{2 + \alpha(1+\lambda)} (k_1 + k_2) + \frac{\alpha(1 - \lambda)}{2 + \alpha(1+\lambda)} k_3 \right) \\
        & = & \beta\lambda ( \gamma_1 k_1 + \gamma_2 k_2 + \gamma_3 k_3 ),
\end{eqnarray*}
where $\gamma_1$, $\gamma_2$, and $\gamma_3$ are non-negative and sum to unity. The convex combination $\gamma_1 k_1 + \gamma_2 k_2 + \gamma_3 k_3$ defines some $k_4 \in K$ such that
\[
  v
    ~  =  ~ (1 + \beta\lambda) x - \beta\lambda k_4
    ~  =  ~ x + \beta\lambda (x - k_4)
    ~ \in ~ x + \beta\lambda (x - K).
\]
To complete the proof, it remains to show that $v \in x + \beta \lambda(K - x)$. This follows by a symmetrical argument, which we omit. (This other case holds under the weaker condition that $c \ge \frac{3-\lambda}{1+\lambda}$.)
\end{proof}

\PackingCovering*

\begin{proof}
By definition, the $E''$ ellipsoids define a packing. To show that the $E'$ ellipsoids define a covering of $K_{\eps}$, take any point $y \in K_{\eps}$. By the maximality of $X_{\eps}$, there exists $x \in X_{\eps}$ such that $E''(y) \cap E''(x) \neq \emptyset$. By applying Lemma~\ref{exp-con-ellipse.lem} (with $\lambda = \lambda_p$, $\alpha = 1$, and $\beta = (3 + \lambda_p)\sqrt{d}/(1 - \lambda_p)$) we have 
\[
    y 
        ~ \in       ~ E''(y) 
        ~ =         ~ E^{\lambda_p}(y) 
        ~ \subseteq ~ E^{\beta\lambda_p}(x).
\]
Clearly, $\lambda_p < \inv{5}$, from which we obtain 
\[
	\beta\lambda_p
		~ = ~ \frac{(3+\lambda_p)\sqrt{d}}{1 - \lambda_p}\lambda_p 
		~ < ~ \frac{(3 + 1/5)\sqrt{d}}{4/5} \inv{8\sqrt{d}}
		~ = ~ \inv{2}
		~ = ~ \lambda_c.
\]
Therefore, $E^{\beta\lambda_p}(x) \subset E^{\lambda_c}(x) = E'(x)$. Thus, every point of $K$ is contained in the interior of some $E'$ ellipsoid, implying that they form a covering of $K$. Finally, since $\lambda_c < 1$ every ellipsoid $E'(x)$ lies within $K$, and hence so does their union.
\end{proof}

\APMOverlapBound*

\begin{proof}
Consider any point $x \in X_{\eps}$, and let $\Gamma(x)$ denote its neighbors in $G_{\eps}(K)$, that is $\Gamma(x) = \{y \in X_{\eps} : E'(x) \cap E'(y) \neq \emptyset\}$. Recall that $E''(x) = E^{\lambda_p}(x)$ and $E'(x) = E^{\lambda_c}(x)$, where $\lambda_p = \inv{8\sqrt{d}}$ and $\lambda_c = \inv{2}$. Let $\lambda_0 = 3 \sqrt{d}$.

By applying Lemma~\ref{exp-con.lem} (with $\lambda = \lambda_c$, $\alpha = \frac{\lambda_p}{\lambda_c} = \inv{4\sqrt{d}}$, and $\beta = \frac{2 + \alpha(1+\lambda_c)}{1-\lambda_c}\sqrt{d}$), we have 
\begin{equation}
	E^{\lambda_p}(y) 
		~ =         ~ E^{\alpha \lambda_c}(y) 
		~ \subseteq ~ E^{\beta\lambda_c}(x)
		~ \subset   ~ E^{\lambda_0}(x),
    \label{neighbor-inclusion.eqn}
\end{equation}
where the last inclusion follows because 
\[
    \beta\lambda_c
        ~ = ~ \left(\frac{2 + \alpha(1+\lambda_c)}{1-\lambda_c}\sqrt{d}\right) \lambda_c
        ~ = ~ 2\sqrt{d} + \frac{3}{8}
        ~ < ~ 3 \sqrt{d}
        ~ = ~ \lambda_0.
\]
Swapping the roles of $x$ and $y$ in the application of Lemma~\ref{exp-con.lem}, we have $E^{\lambda_p}(x) \subset E^{\lambda_0}(y)$. Recalling that $\vol(C^\lambda) = \lambda^d \cdot \vol(C)$, when $C^\lambda$ is a uniform $\lambda$-factor scaling of any bounded full-dimensional set $C \subset \RE^d$, this implies that
\begin{equation}
    \vol(E^{\lambda_p}(y))
        ~ = ~ \left(\frac{\lambda_p}{\lambda_0}\right)^{\kern-2pt d} \vol(E^{\lambda_0}(y))
        ~ > ~ \left(\frac{\lambda_p}{\lambda_0}\right)^{\kern-2pt d} \vol(E^{\lambda_p}(x))
        ~ = ~ \left(\frac{\lambda_p}{\lambda_0}\right)^{\kern-2pt 2 d} \vol(E^{\lambda_0}(x)).
        \label{neighbor-vol.eqn}
\end{equation}
By Eq.~\eqref{neighbor-inclusion.eqn}, the packing ellipsoid of every neighbor $y$ of $x$ lies within $E^{\lambda_0}(x)$, and by Eq.~\eqref{neighbor-vol.eqn} the volume of each packing ellipsoids is at least a constant fraction (namely, $(\lambda_p/\lambda_0)^{2 d}$) of the volume of $E^{\lambda_0}(x)$. Since these packing ellipsoids are pairwise disjoint, the total number of neighbors of $x$ is a constant. 
\end{proof}

\APMConstruction*

\begin{proof}
First, we transform $K$ into canonical form, and replace it with an $\frac{\eps}{2}$-approximation $K'$ of itself. This can be done in $O(m + 1/\eps^{O(d)})$ time, so that $K'$ is bounded by $O(1/\eps^{(d-1)/2})$ halfspaces (see, e.g., \cite{AFM18a}). We then build the data structure to solve APM queries to an accuracy of $(\eps/2)$, so that the total error is $\eps$. 

To compute the set $X_{\eps}(K)$, recall that the Hausdorff distance between $K$ and $K_{\eps}$ is $\eps$. As observed at the end of Section~\ref{delone.sec}, $E''(x)$ contains a Euclidean ball of radius $\Omega(\lambda_p \eps) = \Omega(\eps)$. We restrict the points of $X_{\eps}(K)$ to come from the vertices of a square grid whose side length is half this radius. Since $K$ is in canonical form, it suffices to generate $O(1/\eps^{O(d)})$ grid points. By decreasing the value of $\eps$ slightly (by a constant factor), it is straightforward to show that any Delone set can be perturbed so that its centers lie on this grid.

Each Macbeath ellipsoid can be computed in time linear in the number of halfspaces bounding $K'$, which is $O(1/\eps^{O(d)})$~\cite{ChM96}. The maximal set is computed by brute force, repeatedly selecting a point $x$ from the grid, computing $E''(x)$, and marking the points of the grid that it covers until all points interior to $K$ are covered. Finally, to build the graph $G_\eps(K)$, we find the ellipsoids overlapping $E'(x)$ for each selected grid point $x$. For each selected grid point $y$ within a constant factor expansion of $E'(x)$, we test whether $E'(x) \cap E'(y) \neq \emptyset$. If so, we add the edge $(x, y)$ in $G_\eps(K)$.

The overall running time is dominated by the product of the number of grid points and the $O(1/\eps^{O(d)})$ time to compute each Macbeath ellipsoid.
\end{proof}

\end{document}